\documentclass[letterpaper, 10 pt, conference]{ieeeconf}

\IEEEoverridecommandlockouts
\usepackage{graphics}
\usepackage{epsfig}
\usepackage{mathptmx}
\usepackage{times}
\usepackage{amsmath}
\usepackage{amssymb}
\usepackage{color}
\usepackage{subfig}
\usepackage{mathrsfs}
\usepackage{booktabs}
\usepackage{epstopdf}
\usepackage{multirow}
\usepackage{fancyhdr}

\newtheorem{theorem}{Theorem}[section]

\newtheorem{defn}{Definition}[section]

\newtheorem{remark}{Remark}[section]

\title{\LARGE \bf Risk Identification of Power Transmission System with Renewable Energy}

\author{Chao~Zhai, Gaoxi~Xiao, Hehong Zhang and Tso-Chien Pan \thanks{Chao Zhai, Gaoxi Xiao, Hehong Zhang and Tso-Chien Pan are with Institute of Catastrophe Risk Management, Nanyang Technological University, 50 Nanyang Avenue, Singapore 639798. They are also with Future Resilient Systems, Singapore-ETH Centre, 1 Create Way, CREATE Tower, Singapore 138602. Chao Zhai, Gaoxi Xiao and Hehong Zhang are also with School of Electrical and Electronic Engineering, Nanyang Technological University. Corresponding author: Gaoxi Xiao. Email: {\tt\small egxxiao@ntu.edu.sg}}
}

\begin{document}

\maketitle
\thispagestyle{empty}
\pagestyle{empty}

%%%%%%%%%%%%%%%%%%%%%%%%%%%%%%%%%%%%%%%%%%%%%%%%%%%%%%%%%%%%%%%%%%%%%%%%%%%%%%%%
\begin{abstract}
This paper aims to investigate the risk identification problem of power transmission system that is integrated with renewable energy sources. In practice, the fluctuation of power generation from renewable energy sources can lead to the severe consequences to power transmission network. By treating the fluctuation of power generation as the control input, the risk identification problem is formulated with the aid of optimal control theory. Thus, a control approach is developed to identify the fluctuation of power generation that results in the worst-case cascading failures of power systems. Theoretical analysis is also conducted to obtain the necessary condition for the worst fluctuations of power generation. Finally, numerical simulations are implemented on IEEE 9 Bus System to demonstrate the effectiveness of the proposed approach.
\end{abstract}

\section{INTRODUCTION}\label{sec:int}
In the past decades, renewable energy has attracted much interest and attention from academia, industries and governments due to its great advantages. Compared to traditional energy (e.g., coal, petroleum and natural gas), one major advantage of renewable energy lies in its sustainability, which implies that it will never be exhausted for human beings. Moreover, it can provide clean energy without aggravating greenhouse effects and global warming \cite{pan11}. Therefore, a lot of efforts have been taken to exploit renewable energy for generating electric power as a supplement to the conventional power stations \cite{car06}. Nevertheless, most renewable energy largely depends on the weather condition (e.g., wind speed and light intensity) to generate the power electricity, which may bring about the fluctuation of power generation and even trigger the cascading failure of power grids.

It has been widely investigated for the effect of fluctuating renewable-energy sources on the operation of power grids \cite{mat09,lun03,con12}. For instance, comparative analyses of seven technologies are presented to integrate fluctuating renewable energy sources into the electricity supply in order to identify the most fuel-efficient and least-cost technologies \cite{mat09}. In addition, \cite {lun03} discusses and analyzes different national strategies for solving the problem of ``surplus production" from a fluctuating renewable-energy source. And \cite{con12} examines the role of large-scale energy storage (e.g., pumped hydro) in the integration of fluctuating renewable energy (e.g., wind energy). Actually, the fluctuation of renewable energy sources might trigger cascading failures of power grids, which normally occur as a result of initial contingencies and deteriorate with a sequence of branch outages \cite{lu96}. The initiating events of cascading failure may result from the outage of a certain transmission line or the fluctuation of injected power on buses. Some approaches are developed to identify the initiating events that trigger cascading failure, and they include the random chemistry algorithm \cite{epp12}, nonlinear programming method \cite{kim16} and optimal control approach \cite{czm17}. Regarding the optimal control approach, the main idea is to identify critical risks of power grids by treating the initial disturbances as control inputs \cite{czm17,cziw17,zhang17}. To be specific, \cite{czm17} develops an identification algorithm for determining the admittance changes on selected branches that can cause the worst-case cascading failure of power systems. And \cite{cziw17} investigates the problem of identifying initial disturbances on branches by comparing three different cascading failure models. In addition, \cite{zhang17} proposes a numerical algorithm to identify both the critical branches and the admittance changes in power transmission systems. To the best of our knowledge, it is largely unexplored to identify the severe fluctuation of injected power on buses of power grids from a control perspective. Thus, this work aims to look into the worst fluctuations of power generation from renewable energy sources in the framework of optimal control theory. In short, the main contributions of this paper
are summarized as follows:
\begin{enumerate}
  \item Formulate the risk identification problem of power transmission system integrated with renewable energy using optimal control theory.
  \item Provide theoretical results on the necessary condition of the worst fluctuations of power generation from renewable energy sources.
  \item Validate the proposed approach of risk identification on the IEEE 9 Bus System.
\end{enumerate}
The remainder of this paper is organized as follows. Section \ref{sec:prob} presents the problem formulation of identifying the fluctuation of power generation from renewable energy sources. Section \ref{sec:ana} provides theoretical results on the necessary condition of the worst fluctuations on generator buses. Numerical simulations on IEEE $9$ Bus System are carried out in Section \ref{sec:sim}. Finally, we draw a conclusion and discuss the future work in Section \ref{sec:con}.

\section{PROBLEM FORMULATION}\label{sec:prob}
\begin{figure}
\scalebox{0.05}[0.05]{\includegraphics{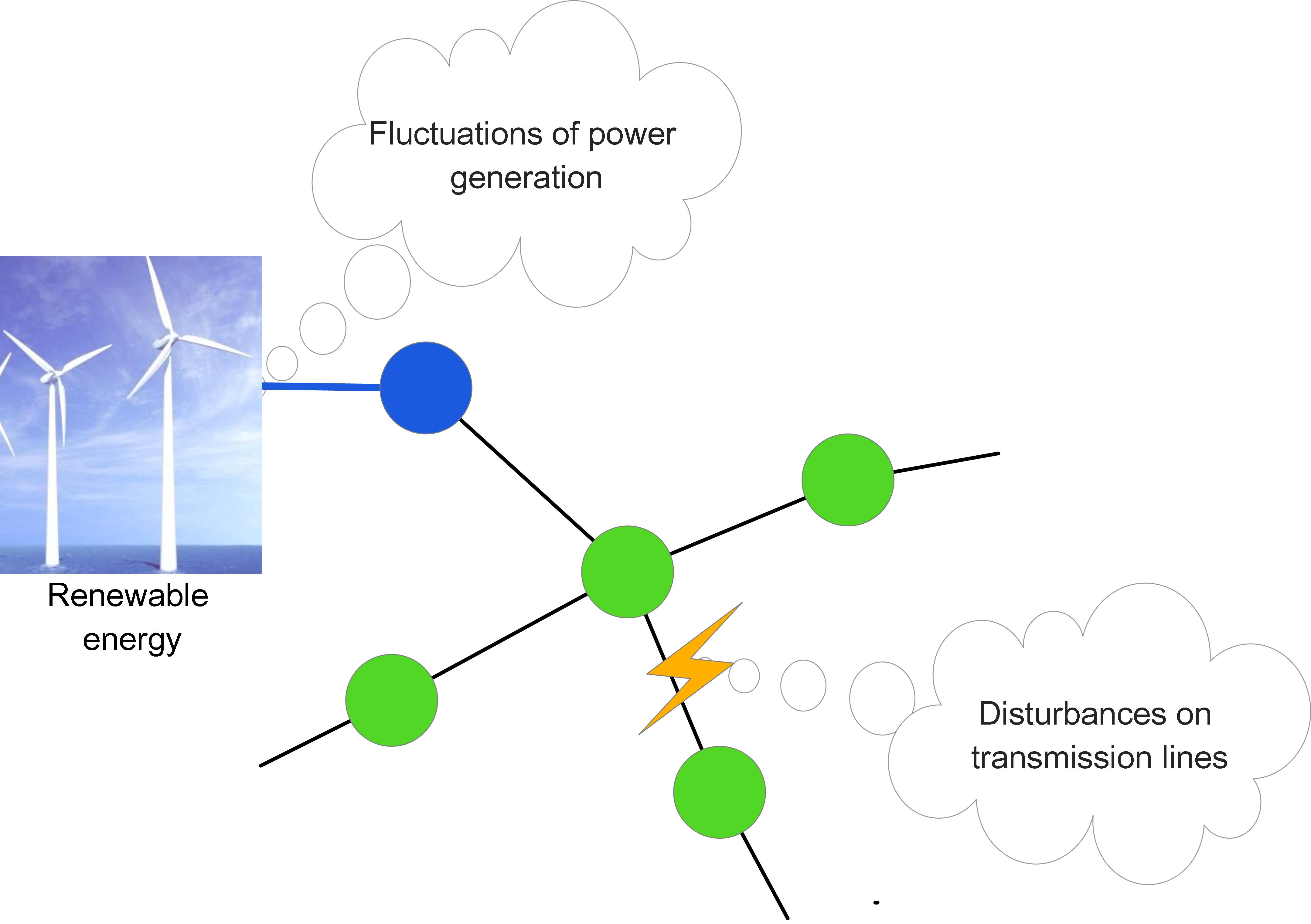}}\centering
\caption{\label{pts} Disturbances in power transmission system integrated with renewable energy sources. The blue ball denotes the generator bus that receives the power supply from renewable energy sources, and the green balls represent load buses that consume electric power. Two types of disturbances are taken into account in the power transmission system. One type of disturbances comes from the transmission line (e.g., branch outages due to the lightening or other contingencies), and the other is related to the fluctuations of injected power on generator buses.}
\end{figure}
Figure \ref{pts} presents an illustration of power transmission system that is composed of power generation from renewable energy sources (e.g., wind power and solar power) and power network. In the power system, The electric power is generated by power plants using renewable energy (e.g., wind energy and solar energy) and power stations with traditional energy (e.g., coal and natural gas). In practice, the fluctuations of power generation from renewable energy sources have great impact on the stability of power systems. Therefore, it is vital to identify the worst fluctuations that can cause the massive disruptions of power transmission system.

In this section, we will formulate the risk identification problem using optimal control by treating the fluctuations of power generation as control inputs. This is because power transmission system can be regarded as an optimal control system (see Table \ref{ana}). Specifically, the current, voltage, power and impedance can be described by the states in optimal control system, and the fluctuations of power generation from renewable energy sources can be treated as control inputs. In addition, the state equation describes the evolution of cascading failure process in power transmission systems, and the physical laws in electrical power systems (e.g., Kirchhoff Circuit Laws) are equivalent to state constraints in the control system. Moreover, the cost function can be designed to quantify the disruption level of power transmission system due to the fluctuations of power generation from renewable energy sources.

\begin{table}
\caption{\label{ana} Treating power transmission system as optimal control system.}
\begin{center}
    \begin{tabular}{| l | l |}
    \hline
    \bf{Power Transmission System} & \bf{Optimal Control System} \\ \hline
    Current, voltage, power, impedance & State variables                \\ \hline
    Fluctuations of power generation &   Control inputs                   \\ \hline
    Cascading process & State equation                 \\ \hline
    Physical laws & State constraints                  \\ \hline
    Worst-case disruptions & Cost function             \\ \hline
    \end{tabular}
\end{center}
\end{table}

\subsection{The cascading model of branch outage}
First of all, we present the cascading model of branch outage that describes the dynamics of power transmission system while suffering from the fluctuations of injected power on generator buses. For the high-voltage power transmission system, it is suggested that the direct current (DC) power flow equation can be adopted to approximate the power flow on the transmission line. For a high-voltage power transmission network with $n$ branches and $n_b$ buses, the DC power flow equation is given by
\begin{equation*}\label{dc_power}
P_b=A^Tdiag(Y_p)A\theta,
\end{equation*}
where $P_b\in R^{n_b}$ refers to the vector of injected power on buses, and $A\in R^{n\times n_b}$ is the branch-to-bus incidence matrix. $Y_p$ represents the vector of branch admittance, and $\theta$ denotes the vector of voltage angle on each bus. Then the solution to the above DC power flow equation is given by
$$
\theta=(A^Tdiag(Y_p)A)^{-1^*}P_b,
$$
where the operator $-1^*$ is a type of pseudo inverse as defined in \cite{cz17} (see Definition $5.1$ in the Appendix). Thus, the vector of power flow on $n$ branches can be computed by
\begin{equation*}
P=diag(Y_p)A(A^T diag(Y_p)A)^{-1^*}P_b,
\end{equation*}
where $P=(P_{ij})\in R^n$ with $i,j\in I_{n_b}=\{1,2,...,n_b\}$(see Subsection $B$ in the Appendix for the detail).

When the power flow on a branch exceeds the given threshold, the branch is tripped by the circuit breaker. The above protection mechanism can be described by the approximate function $g(P_{ij},c_{ij})$ as defined in \cite{czm17,cz17}, where $c_{ij}$ represents the threshold of power flow on the branch connecting Bus $i$ to Bus $j$, and $P_{ij}$ refers to the power flow on the corresponding branch. It is demonstrated  that $g(P_{ij},c_{ij})$ approximates to step function as the parameter $\sigma$ in it gradually increases.
In particular, $g(P_{ij},c_{ij})$ is differentiable with respect to $P_{ij}$, which enables us to apply optimal control theory for identifying the worst fluctuation of power generation from renewable energy sources. Then the dynamics of branch outage can be described as follows
\begin{equation}\label{outage}
    Y^{k+1}_p=G(P^k)\circ Y^{k}_p, \quad k=0,1,...,m-1
\end{equation}
where the vector $G(P^k)$ is defined as
$$
G(P^k)=\left(g(P^k_{i_1j_1},c_{i_1j_1}),g(P^k_{i_2j_2},c_{i_2j_2}),...,g(P^k_{i_nj_n},c_{i_nj_n})\right)^T\in R^n.
$$
The symbol $\circ$ denotes the Hadamard product, and $Y^{k}_p$ and $P^k$ refer to the vector of branch susceptance and the vector of power flow on branches at the $k$-th cascading step, respectively. To be precise, the cascading step is defined as one topology change of power networks due to branch outages. In addition, $P^k$ can be computed by
\begin{equation}\label{uk}
P^k=diag(Y^k_p)A(A^T diag(Y^k_p)A)^{-1^*}(P_b+\Lambda U^k),
\end{equation}
where $U^k$ denotes the vector of power fluctuations on buses with respect to the original vector $P_b$ at the $k$-th cascading step, and $diag(Y^k_p)$ allows to obtain the diagonal matrix with the diagonal entries being the elements of the vector $Y_p^k$ in sequence.
In addition, $\Lambda$ is a diagonal matrix with the diagonal elements being $0$ or $1$, and it allows us to select the desired bus by assigning $1$ to the corresponding diagonal element.

\begin{figure}
\scalebox{0.05}[0.05]{\includegraphics{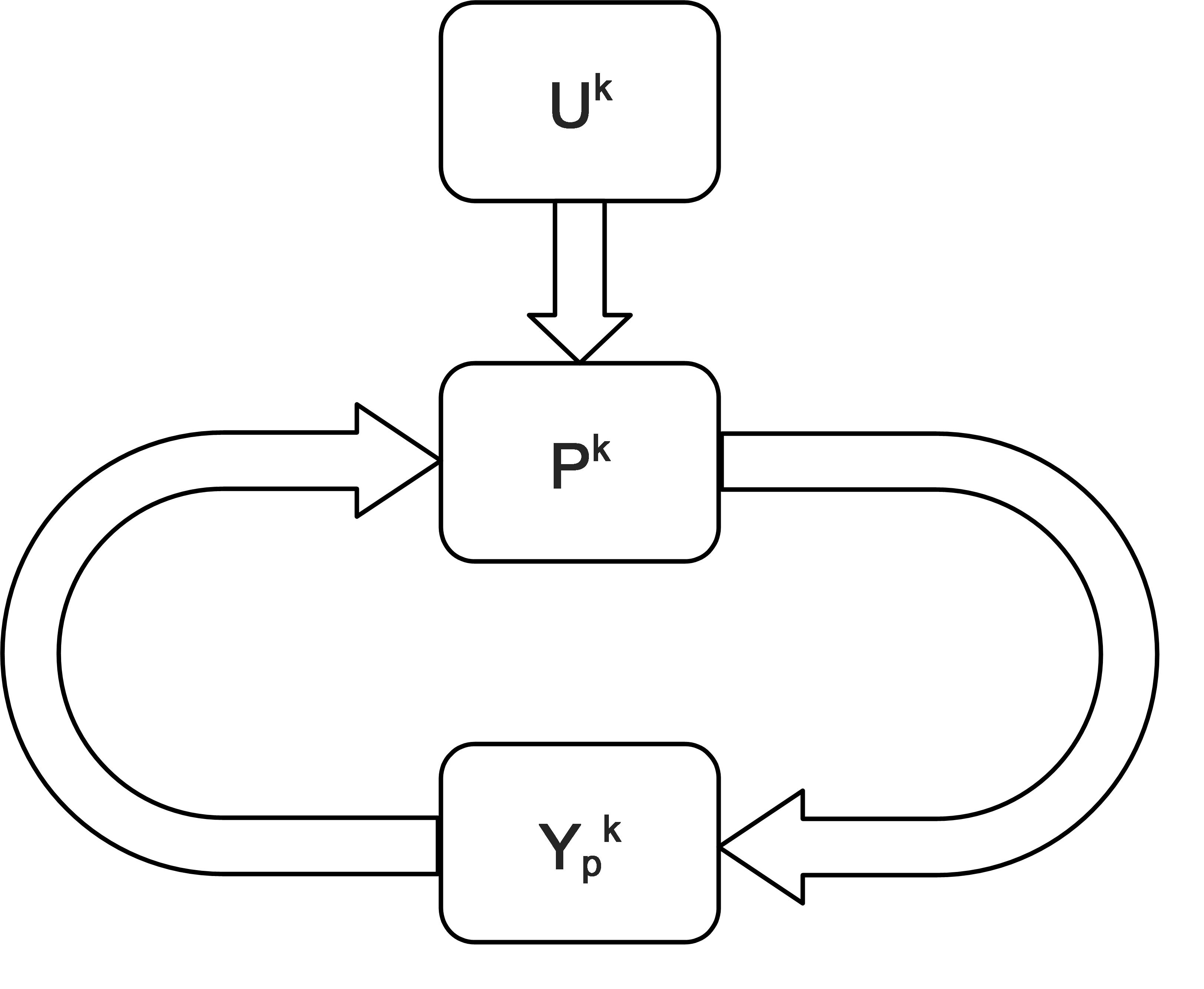}}\centering
\caption{\label{cas} Illustration on the cascading model of branch outage. The input $U^k$ is the vector of power fluctuations on buses at the $k$-th cascading step. $Y^{k}_p$ and $P^k$ represent the vector of branch susceptance and the vector of power flow on branches at the $k$-th cascading step, respectively.}
\end{figure}

Essentially, the cascading model of branch outage is obtained by integrating Equation (\ref{outage}) with Equation (\ref{uk}). Figure \ref{cas} presents the schematic sketch on the cascading model of branch outage. To be specific, the fluctuations of power generation from renewable energy sources (i.e., control inputs $U^k$) give rise to the branch overloads (i.e., the change of power flow $P^k$), which results in the branch outages and the topology change of power networks (i.e., the change of branch susceptance $Y_p^k$). In reserve, the change of network topology leads to the redistribution of power flow, which may result in the further branch outages of power grids. The above process proceeds until the power flow on each branch is less than its given threshold.

\subsection{Optimization formulation}
By treating the branch susceptance as the state in optimal control system, the cascading model of branch outage can serve as the state equation in optimal control theory. As mentioned before, our goal is to identify the fluctuations of injected power on buses that can result in the massive disruptions of power grids. Thus, it is necessary to come up with a performance index that is able
to quantify the disruption level of power grids caused by the fluctuations of injected power on buses. This can be achieved by the optimal control formulation as follows.
\begin{equation}\label{cost}
\min_{U^k}J(Y_p^m,U^k)
\end{equation}
with the cost function
\begin{equation}\label{cost_fun}
J(Y_p^m,U^k)=\|Y_p^m\|^2+\epsilon\sum_{k=0}^{m-1}\|U^k\|^2,
\end{equation}
where the cost function $J(Y_p^m,U^k)$ can be used to quantify the disruptions of power grids, and $m$ denotes the total number of cascading steps. There are two terms in the cost function. Specifically, the first term quantifies the connectivity of power networks at the end of cascading failures, and the second term describes the magnitude of power fluctuations on selected generator buses (i.e., control energy in optimal control system) at the different cascading steps. It is worth pointing out that the tunable parameter $\epsilon$ can be small enough so that the first term dominates the cost function, and thus more efforts are taken by the optimization algorithm to minimize the value of the first term.
\begin{remark}
The parameter $\epsilon$ can be assigned differently at the different cascading step, which allows us to control the fluctuation magnitude of power generation. For the larger $\epsilon$, the smaller $U^k$ (i.e., the weaker fluctuation) is required to result in the worst disruptions of power transmission system.
\end{remark}

\section{Theoretical Analysis}\label{sec:ana}
This section presents theoretical results on the optimal control problem (\ref{cost}), which has been formulated in Section \ref{sec:prob}. The following theorem provides the necessary condition for identifying the worst fluctuations of injected power on buses.
\begin{theorem}\label{sysAE}
The necessary condition for the optimal control problem (\ref{cost}) is obtained by solving the following system of algebraic equations
\begin{equation}\label{con_sys}
\left\{
  \begin{array}{ll}
    P^k=diag(Y^k_p)A(A^T diag(Y^k_p)A)^{-1^*}(P_b+\Lambda U^k) \\
    Y_p^{k+1}=G(P^k)\circ Y_p^k, \quad \quad k=0,1,...,m-1.
  \end{array}
\right.
\end{equation}
where the control input $U^k$ is given by
$$
U^k=-\frac{1}{\epsilon}\left(\frac{\partial Y_p^{k+1}}{\partial U^k}\right)^T\prod_{s=0}^{m-k-2}\frac{\partial Y_p^{m-s}}{\partial Y_p^{m-s-1}}\cdot Y_p^m
$$
with the term
$$
\frac{\partial Y_p^{k+1}}{\partial U^k}=diag(Y_p^k)\frac{\partial G(P^k)}{\partial P^k}diag(Y^k_p)A(A^T diag(Y^k_p)A)^{-1^*}\Lambda.
$$
\end{theorem}

\begin{proof}
It follows from the discrete time optimal control theory \cite{fran95} that the necessary conditions for the optimal control problem (\ref{cost}) can be determined by the following four equations
\begin{equation*}\label{cond_st}
Y_p^{k+1}=G(P^k)\circ Y_p^{k}
\end{equation*}

\begin{equation}\label{cond_u}
\left(\frac{\partial Y_p^{k+1}}{\partial U^k}\right)^T\lambda_{k+1}+\epsilon\cdot\frac{\partial\|U^k\|^2}{\partial U^k}=0
\end{equation}

\begin{equation}\label{cond_y}
\lambda_k=\left(\frac{\partial Y_p^{k+1}}{\partial Y_p^k}\right)^T\lambda_{k+1}+\epsilon\cdot\frac{\partial\|U^k\|^2}{\partial Y_p^k}
\end{equation}

\begin{equation}\label{cond_final}
\lambda_m=\frac{\partial \|Y_p^m\|^2}{\partial Y_p^m}
\end{equation}
By solving (\ref{cond_u}), we have
\begin{equation}\label{cond_uk}
U^k=-\frac{1}{2\epsilon}\left(\frac{\partial Y_p^{k+1}}{\partial U^k}\right)^T\lambda_{k+1}.
\end{equation}
Since
$$
\frac{\partial\|U^k\|^2}{\partial Y_p^k}=0
$$
in Equation (\ref{cond_y}), we can obtain
\begin{equation*}\label{cond_lamb}
\lambda_k=\left(\frac{\partial Y_p^{k+1}}{\partial Y_p^k}\right)^T\lambda_{k+1}
\end{equation*}
with the terminal condition $\lambda_m=2Y_p^m$ according to (\ref{cond_final}). Therefore, we obtain
\begin{equation}\label{cond_lambd}
\lambda_{k+1}=2\prod_{s=0}^{m-k-2}\frac{\partial Y_p^{m-s}}{\partial Y_p^{m-s-1}}\cdot Y_p^m.
\end{equation}
Combining (\ref{cond_uk}) and (\ref{cond_lambd}), we have
\begin{equation*}\label{cond_ukh}
U^k=-\frac{1}{\epsilon}\left(\frac{\partial Y_p^{k+1}}{\partial U^k}\right)^T\prod_{s=0}^{m-k-2}\frac{\partial Y_p^{m-s}}{\partial Y_p^{m-s-1}}\cdot Y_p^m
\end{equation*}
with the term
$$
\frac{\partial Y_p^{k+1}}{\partial U^k}=diag(Y_p^k)\frac{\partial G(P^k)}{\partial P^k}diag(Y^k_p)A(A^T diag(Y^k_p)A)^{-1^*}\Lambda.
$$
This completes the proof.
\end{proof}

\begin{remark}
By substituting the control inputs
$$
U^k=-\frac{1}{\epsilon}\left(\frac{\partial Y_p^{k+1}}{\partial U^k}\right)^T\prod_{s=0}^{m-k-2}\frac{\partial Y_p^{m-s}}{\partial Y_p^{m-s-1}}\cdot Y_p^m
$$
into the system (\ref{con_sys}) and replacing the term
$$
\frac{\partial Y_p^{k+1}}{\partial U^k}
$$
in the control input $U^k$ with
$$
diag(Y_p^k)\frac{\partial G(P^k)}{\partial P^k}diag(Y^k_p)A(A^T diag(Y^k_p)A)^{-1^*}\Lambda,
$$
simultaneously, we can obtain a system of algebraic equations with the unknowns $Y_p^k$, $k=1,2,...,m$. The solutions to the above system of algebraic equations enable us to compute the optimal control input $U^k$ at each cascading step, which corresponds to the worst fluctuation of injected power on buses.
\end{remark}

\begin{remark}
In practice, the computation burden of solving the system (\ref{con_sys}) with the unknowns $Y_p^k$ is relatively high. This is because both the dimension of system (\ref{con_sys}) and the number of unknowns are quite large for a large-scale power transmission system. Moreover, the operations of pseudo inverse $-1^{*}$ and partial derivatives in the system of algebraic equations require much computation time for the large power networks. To reduce the computation time of the numerical optimization algorithm, it is necessary to simplify the  system (\ref{con_sys}) by approximating the terms of high computation complexity. For example, the two terms of high computation complexity in the system (\ref{con_sys}) can be approximated as follows.
$$
\frac{\partial G(P^k)}{\partial P^k}\approx\left(
                                             \begin{array}{c}
                                               \frac{[G(P^k+\delta e_1)-G(P^k)]^T}{\delta} \\
                                               \frac{[G(P^k+\delta e_2)-G(P^k)]^T}{\delta} \\
                                               . \\
                                               \frac{[G(P^k+\delta e_n)-G(P^k)]^T}{\delta} \\
                                             \end{array}
                                           \right)^T
$$
and
\begin{equation*}
\begin{split}
\prod_{s=0}^{m-k-2}\frac{\partial Y_p^{m-s}}{\partial Y_p^{m-s-1}}&=\frac{\partial Y_p^m}{\partial Y_p^{k+1}}
\approx \left(
           \begin{array}{c}
           \frac{[Y_p^m(Y_p^{k+1}+\delta e_1)-Y_p^m(Y_p^{k+1})]^T}{\delta} \\
           \frac{[Y_p^m(Y_p^{k+1}+\delta e_2)-Y_p^m(Y_p^{k+1})]^T}{\delta} \\
           . \\
           \frac{[Y_p^m(Y_p^{k+1}+\delta e_n)-Y_p^m(Y_p^{k+1})]^T}{\delta} \\
           \end{array}
         \right)^T,
\end{split}
\end{equation*}
where $e_i\in R^n$ denotes the $n$ dimensional unit vector with the $i$-th element being 1 and all other elements being 0, and the parameter $\delta$ should be sufficiently small.
\end{remark}

\section{NUMERICAL SIMULATIONS}\label{sec:sim}
\begin{figure}
\scalebox{2.4}[2.4]{\includegraphics{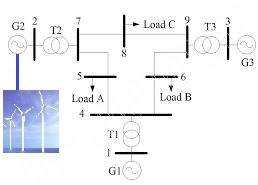}}\centering
\caption{\label{bus9} IEEE 9 Bus System with power supply from renewable energy sources.}
\end{figure}
In this section, the numerical algorithm is adopted to obtain solutions to the system (\ref{con_sys}). In addition, numerical simulations are carried out to validate the proposed optimal control approach to identify the fluctuations of injected power on the selected bus of IEEE 9 Bus System \cite{zim11}. As is shown in Fig.~\ref{bus9}, Bus $2$ is connected to the renewable energy sources (i.e, wind turbines), which leads to the fluctuations of injected power on this generator bus due to the variation of wind speed. The other two generators can obtain the stable power supply, and the fluctuations of injected power on these two buses are negligible. Thus, the elements of the diagonal matrix $\Lambda$ is defined by
$$
\Lambda(i,j)=\left\{
  \begin{array}{ll}
    1, & \hbox{$i=j=2$;} \\
    0, & \hbox{otherwise.}
  \end{array}
\right.
$$
For simplicity, per unit values are adopted with the base value of power
100 MVA in numerical simulations. Other parameters are given as follows: $\epsilon=10$, $m=4$, $\delta=0.01$ and $\sigma=1$ in the approximate function. The threshold vector of power flow on branches is $(1, 1.8, 1, 0.5, 1, 1, 1, 1, 1)$. The function ``fsolve" in Matlab (R2017b) is employed to solve the system of algebraic equations (\ref{con_sys}) and then compute the control input $U^k$ (i.e., fluctuation of power generation) at each cascading step. After implementing the numerical solver for the system (\ref{con_sys}), the control inputs $U^k$, $k\in\{1,2,3\}$ on the selected bus (i.e., fluctuations of injected power on Bus $2$) are identified from Step $1$ to Step $3$ as follows: $-1.18$,  $-1.24$ and $1.46$.
\begin{figure}
\scalebox{0.07}[0.07]{\includegraphics{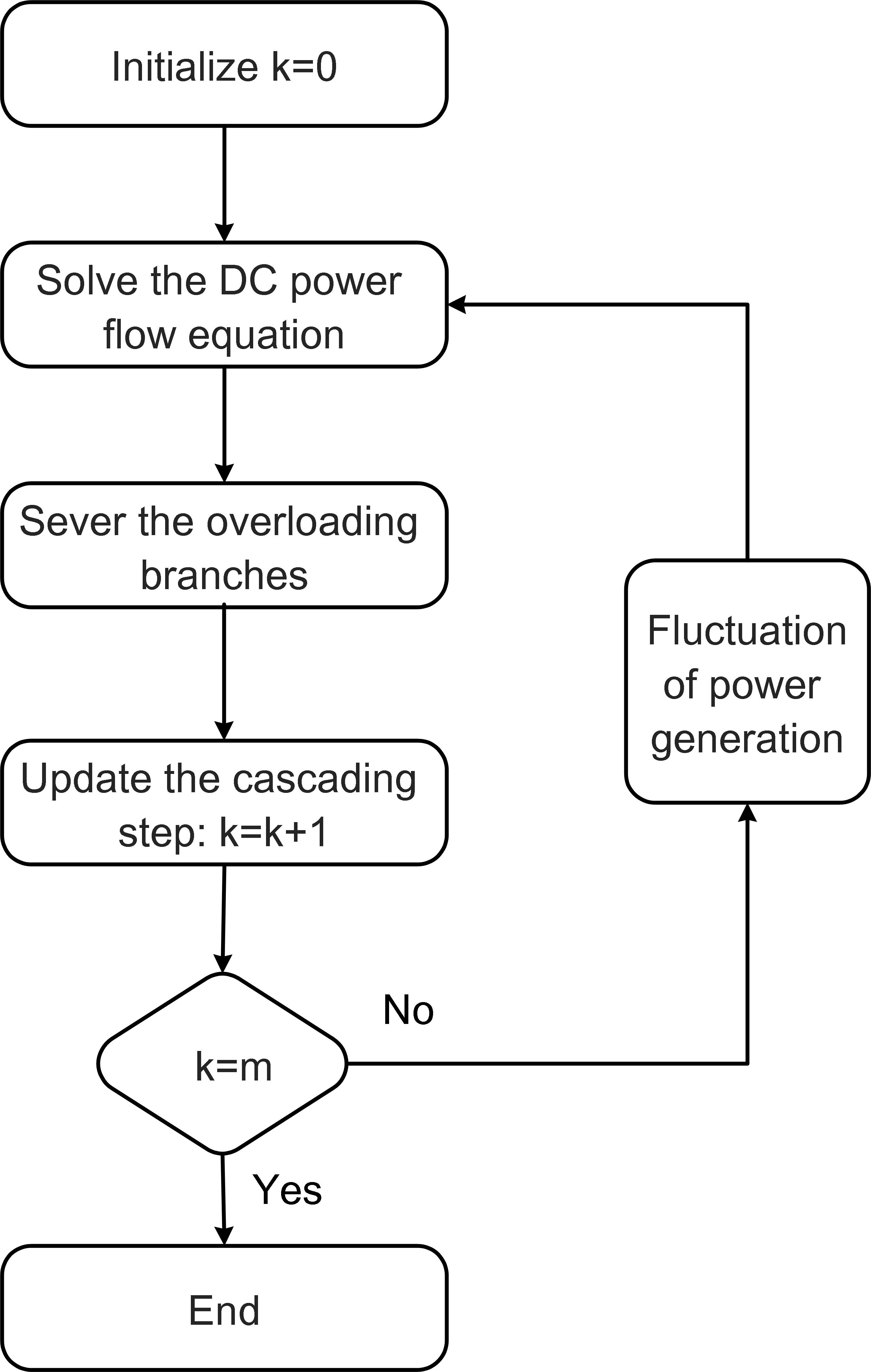}}\centering
\caption{\label{chart} Flow chart of numerical simulations.}
\end{figure}

Figure \ref{chart} demonstrates the process of numerical simulations to validate the proposed approach by adding the fluctuations of power generation on the selected generator bus at each cascading step. Specifically, after initializing the simulation program, the DC power flow equation is solved to obtain the power flow on each branch. Then the overloading branch will be severed if the power flow exceeds the given threshold. Next, the identified fluctuations of power generation (i.e, control inputs $U^k$) are added on the selected bus after updating the cascading step. This will cause the redistribution of power flow on branches and further branch outages in the next round. When the cascading process reaches the $m$-th step, the simulation comes to an end. Figure \ref{flu} presents the evolution of cascading failure on IEEE 9 Bus System by adding the identified control inputs (i.e., fluctuations of injected power) on Bus $2$ at each cascading step. As we can see, the power network is separated into $7$ isolated islands after 4 cascading steps. Note that three buses (i.e., Bus 4, Bus 6 and Bus 9) are connected through two branches without any power flow due to the absence of power supply from generator buses. Moreover, the other $6$ buses are all isolated from each other and fail to survive as well. This implies that the power network is completely damaged by the fluctuations of power generation from renewable energy sources, which are identified by the proposed approach.

\begin{figure}
\scalebox{0.72}[0.72]{\includegraphics{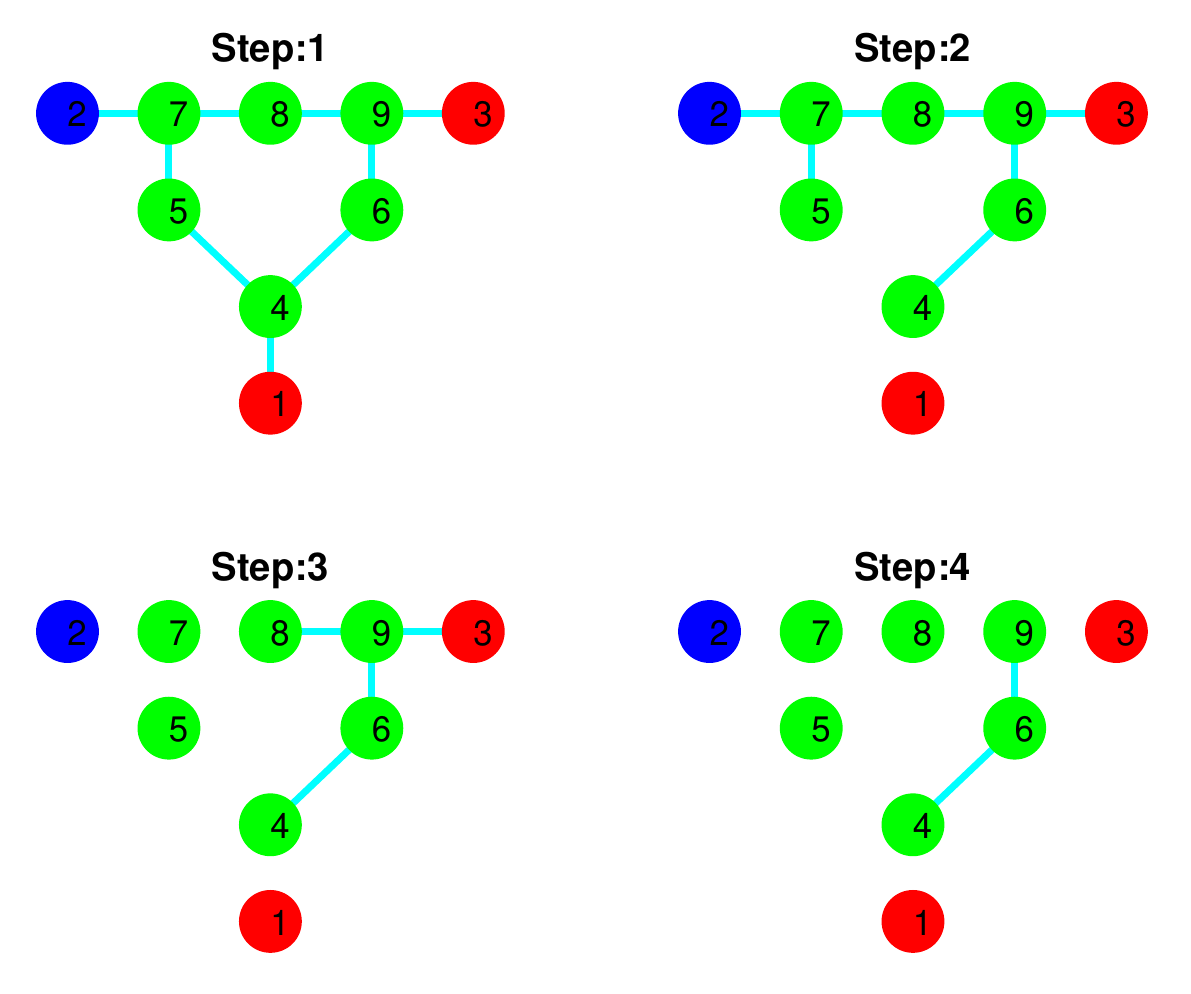}}\centering
\caption{\label{flu} Cascading process of power transmission system with the fluctuations of injected power on Bus $2$. The green balls denote load buses, and the red ones refer to generator buses that can receive the stable power supply. The blue ball (i.e., Bus $2$) is also a generator bus, and it is connected to renewable energy sources that may give rise to the fluctuation of power generation.}
\end{figure}

\section{CONCLUSIONS}\label{sec:con}
In this paper, we investigated the effects of power fluctuations caused by renewable energy sources on the stability and security of power transmission system. The electric power from renewable energy sources was injected into the generator bus of power transmission system.
By regarding the fluctuations of injected power as control inputs, the risk identification problem of power system was formulated with optimal control theory. An approach was developed to identify the severe fluctuations of power generation from renewable energy sources, and it was validated on IEEE $9$ Bus Systems. In the next step, we will extend the proposed approach to power distribution system by taking into account both the AC power flow and the transient process of power grids.

\section*{ACKNOWLEDGMENTS}
This work is partially supported by the Future Resilient Systems Project at the Singapore-ETH Centre (SEC), which is funded by the National Research Foundation of Singapore (NRF) under its Campus for Research Excellence and Technological Enterprise (CREATE) program. It is also partially supported by Ministry of Education of Singapore under contract MOE2016-T2-1-119.

\section*{APPENDIX}\label{app:def}

\subsection{Definition of Operator $-1^{*}$}
Suppose the nodal admittance matrix $Y_b=A^T\texttt{diag}(Y_p)A$ is composed of $q$ isolated subnetworks denoted by $\mathrm{S}_i, i\in I_q=\{1,2,...,q\}$ and each subnetwork $\mathrm{S}_i$ includes $k_i$ buses. Let $V_i=\{i_1,i_2,...,i_{k_i}\}$ denote the set of bus IDs in the subnetwork $S_i$, where $i_1$, $i_2$,..., $i_{k_i}$ denote the bus IDs. Notice that Bus $i_1$ in Subnetwork $\mathrm{S}_i$ is designated as the reference bus. Moreover, the nodal admittance matrix of the $i$-th subnetwork can be computed as $Y_{b, i}=\mathcal{E}_i^TY_{b}\mathcal{E}_i$, $i\in I_q$, where $\mathcal{E}_i=(e_{i_1}, e_{i_2} ,..., e_{i_{k_i}})$. Then the operator $-1^*$ is defined as follows \cite{czm17}.
\begin{defn}
Given the nodal admittance matrix $Y_b$, the operator $-1^*$ is defined by
\begin{equation*}
\begin{split}
\left(Y_{b}\right)^{-1^*}&=\sum_{i=1}^{q}\mathcal{E}_i\left(
           \begin{array}{c}
             0^T_{k_i-1} \\
             I_{k_i-1} \\
           \end{array}
         \right)\left(\mathcal{Y}_{b,i}\right)^{-1}\left(
                                                                                                 \begin{array}{cc}
                                                                                                   0_{k_i-1} & I_{k_i-1} \\
                                                                                                 \end{array}
                                                                                               \right)\mathcal{E}_i^T,
\end{split}
\end{equation*}
where
$$
\mathcal{Y}_{b,i}=\left(
     \begin{array}{cc}
       0_{k_i-1} & I_{k_i-1} \\
     \end{array}
   \right)
Y_{b,i}\left(
           \begin{array}{c}
             0^T_{k_i-1} \\
             I_{k_i-1} \\
           \end{array}
         \right).
$$
$I_{k_i-1}$ is the $(k_i-1)\times(k_i-1)$ identity matrix, and $0_{k_i-1}$ denotes the $(k_i-1)$ dimensional column vector with zero elements.
\end{defn}

\subsection{Computation of the power flow}
According to the DC power flow equation
$$
P_b=A^Tdiag(Y_p)A\theta,
$$
we have
$$
\theta=(A^Tdiag(Y_p)A)^{-1^*}P_b,
$$
Thus, the vector of power flow on the branches can be computed by
\begin{equation*}
\begin{split}
P&=diag(Y_p)A\cdot\theta\\
&=diag(Y_p)A(A^Tdiag(Y_p)A)^{-1^*}P_b.
\end{split}
\end{equation*}

\end{document}